\documentclass[conference,letter, 10pt]{IEEEtran}

\usepackage{amsmath}
\usepackage{amsfonts}
\usepackage{amssymb}

\usepackage{theorem}
\usepackage{graphicx}
\usepackage{psfrag}
\usepackage{enumerate}

\usepackage{caption}
\usepackage{subcaption}
\usepackage{color}

\usepackage{cite}
\newtheorem{theorem}{Theorem}
\newtheorem{lemma}{Lemma}

\newtheorem{definition}{Definition}
\newtheorem{remark}{Remark}
\newtheorem{example}{Example}
\newtheorem{claim}{Claim}

\newcommand{\comment}[1]{}

\def\cX{\mbox{$\cal{X}$}}
\def\cY{\mbox{$\cal{Y}$}}
\def\cZ{\mbox{$\cal{Z}$}}
\def\cV{\mbox{$\cal{V}$}}

\def\cE{\mbox{$\cal{E}$}}

\def\cV{\mbox{$\cal{V}$}}
\def\cG{\mbox{$\cal{G}$}}

\title{Broadcast Function Computation with Complementary Side Information  }
\author{
\IEEEauthorblockN{Jithin~Ravi and Bikash~Kumar~Dey}
\IEEEauthorblockA{Department of Electrical Engineering \\Indian Institute of Technology Bombay \\
        {\tt \{rjithin,bikash\}@ee.iitb.ac.in}
}
}

\begin{document}
\maketitle
 \begin{abstract}
 We consider the function computation problem in a three node network
 with one encoder and two decoders. The encoder has access to 
 two correlated sources $X$ and $Y$.  The encoder encodes $X^n$ and $Y^n$
into a message which is given to two decoders. 
Decoder 1 and decoder 2 have access to $X$ and $Y$ respectively, and they want to compute
 two functions $f(X,Y)$ and $g(X,Y)$ respectively using
 the encoded message and their respective side information. 
 We want to find the optimum (minimum) encoding rate under the zero error and 
$\epsilon$-error (i.e. vanishing error) criteria.
 For the special case of this problem with $f(X,Y) = Y$ and $g(X,Y) = X$, 
 we show that the $\epsilon$-error optimum rate is also achievable with 
zero error.
This result extends to a more general `complementary delivery index
coding' problem with arbitrary number of messages and decoders.
 For other functions, we show that the cut-set bound is 
 achievable under $\epsilon$-error if $X$ and $Y$ are binary, or if
the functions are from a special class of `compatible' functions which
includes the case $f=g$.

\end{abstract}

\section{Introduction}
% Information aggregation of distributed data has been investigated from various
% angles in the past~\cite{Giridhar_2006,Rai_2012,Shah_2013}. A computing node in
% a sensor network is typically interested in computing a function of all the
% distributed data.  Collecting all the data at the computing node is not an
% efficient way of communication. Minimum descriptions of the distributed data
% which are sufficient to compute the desired function, have been addressed from
% information theoretic point of view for small networks \cite{Orlitsky_2001,Ma_2011}.  

We consider {\em the broadcast function network} with
complementary side information as shown in Fig.~\ref{Broadcast_network}.
Here, $(X_i,Y_i)$ is an i.i.d. discrete random process with an
underlying probability mass function $p_{XY}(x,y)$.
An encoder encodes $X^n$ and $Y^n$ into a message, which is
given to two decoders. 
Decoder 1 and decoder 2 have side information $X$ and $Y$ respectively, and
want to compute $Z_1=f(X,Y)$ and $Z_2=g(X,Y)$ respectively.
We study this problem under $\epsilon$-error and zero error criteria.
We are interested in finding the optimum broadcast rate in both cases.
%This problem is also of independent interest. 

\begin{figure}[h]
\centering
\includegraphics[width=\columnwidth]{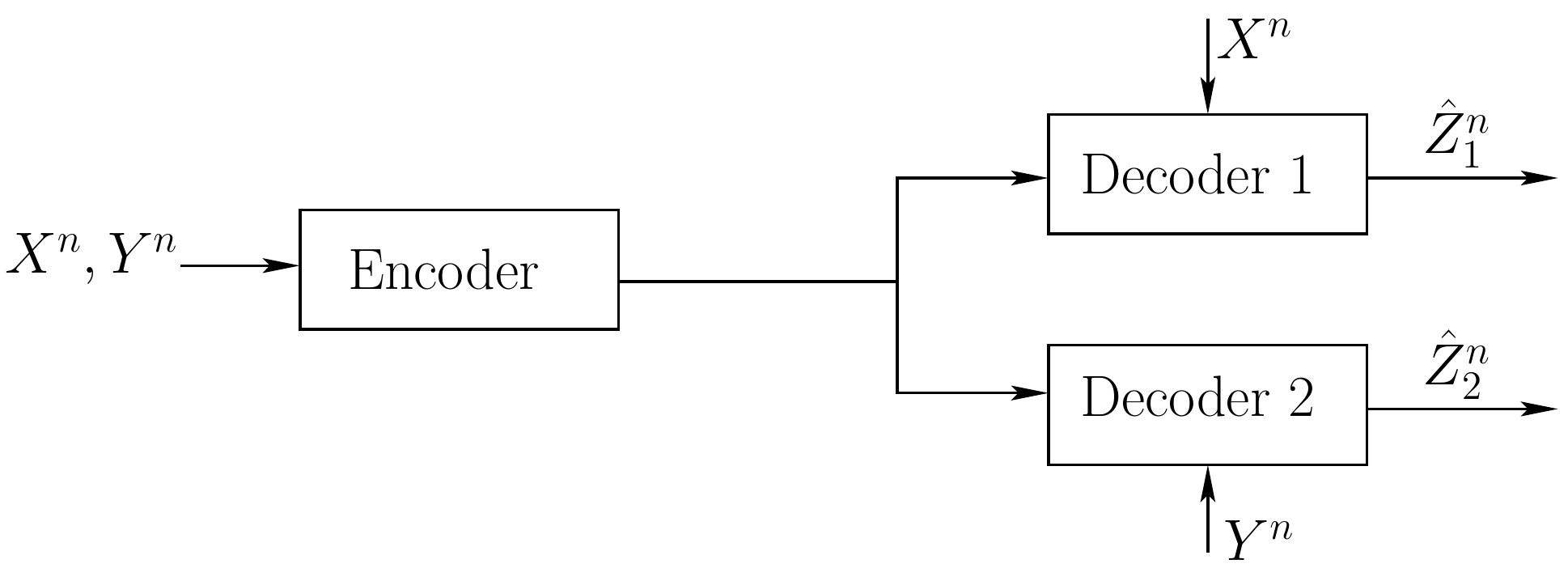}
\caption{Function computation in broadcast function network with complementary side information. Here $Z_1=f(X,Y)$, $Z_2=g(X,Y)$.}
\label{Broadcast_network}
\end{figure}

We first consider a special case of the problem with
$Z_1=Y$ and $Z_2=X$, known as the complementary delivery problem.
This special case 
is an instance of index coding problem with two messages.
This problem has been addressed under noisy broadcast channel in \cite{Tuncel_2006,Wu_2007,Kramer_2007}
for $\epsilon$-error recovery of the messages.
In contrast to their model of independent messages, we consider
correlated messages over a noiseless broadcast channel.
Lossy version of this problem was studied in \cite{Kimura_2007, Timo_2013}. 
For the lossless case, the optimal $\epsilon$-error rate
can be shown to be
$\max\{ H(Y|X), H(X|Y) \}$ using the Slepian-Wolf result.
 We show that this rate is also achievable with zero error.
% We first consider this problem in the context of zero error and show that
% is achievable
% with zero error for variable length coding. 
We then extend this to $n$ random variables 
which can also be considered as a special 
case of the index coding problem. Here,
the server has messages $X_1, \ldots, X_K$ and there are 
$m$ receivers. Each receiver has a subset of $\{ X_1, \ldots, X_K \}$ 
as side information, and all the receivers want 
to recover all the random variables that it does not have access to.
We call this setup as {\em complementary delivery index coding} problem.
Cut-set bound in this case can be shown to be achievable
for $\epsilon$-error using the Slepian-Wolf result.
We show that this rate is also achievable with zero error.

Next we address the function computation problem shown in 
Fig.~\ref{Broadcast_network}, where each decoder wants to recover 
a function of the messages. For $\epsilon$-error criteria, 
we give a single letter characterization of the optimal broadcast rate
when either 
 (i) $Z_1 = Z_2$, 
 (ii) $X,Y$ are binary random variables, or
(iii) $Z_1,Z_2$ belong to a special class of {\em `compatible'} functions (defined in Section~\ref{sec_problem}).
For zero error criteria with variable length coding, we give single
letter upper and 
lower bounds for the optimal broadcast rate.

In contrast to correlated messages in our model, most work on index coding 
consider independent messages. On the other hand, in index coding problems 
in general,
each receiver wants to recover an arbitrary subset of the messages.
The goal is to minimize the broadcast rate of the message sent by the server 
(see \cite{Yossef_2011}-\cite{Maleki_2014} and references therein).
For correlated sources, index coding problem has been studied for 
$\epsilon$-error
where the receivers demand their messages to be decoded with $\epsilon$-error (see for example \cite{Miyake_2015}). 
They gave an inner bound, and showed
that it is tight for three receivers.
To the best of our knowledge, index coding problem has not been considered 
for correlated sources with zero error.
When the sources are independent and uniformly distributed, it was shown that 
the optimal rate for zero error and $\epsilon$-error are the same \cite{Langberg_2011}.
Our result extends this 
to correlated sources with arbitrary distribution
in the specific case of complementary delivery.
The technique followed in \cite{Langberg_2011} does not directly extend
to correlated sources.

% We address the index coding problem in the context of variable length coding 
% where each terminal wants to decode its message of interest with zero error.
% We consider a special case of index coding where 

% The function computation version of index coding problem has been recently 
% proposed in the name of
% functional index coding problem in \cite{Gupta_2016}.
% In this problem, each receiver has a set of functions of meesages
% and demands another set of functions of the messages. All the messages are assumed to be independent 
% and bounds on the optimal rate has been presented in there.

The paper is organized as follows.
In Section~\ref{sec_problem}, we present our problem formulation and some
definitions.  We provide the main results of the paper in 
Section~\ref{sec_results}. Proof of the results are presented in 
Section~\ref{sec_proofs}. 
%We conclude the paper in Section~\ref{sec_conclusion}.
%\JR{Missing proofs in the paper can be found in the full paper \cite{}}

\section{Problem formulation and definitions}
\label{sec_problem}
\subsection{Problem formulation: function computation}

There are one encoder and two decoders for the function computation 
problem shown in Fig~\ref{Broadcast_network}.
% The problem formulation for zero error function computation with
% variable length coding is as follows.
% On observing $(X^n,Y^n)$, encoder broadcasts a message $M$ to decoder 1 and 2.
% After receiving the message, decoders 1 and 2 compute
% $f(x_i,y_i)$ and $g(x_i,y_i)$ respectively using their respective side informations.
A $(2^{nR},n)$ code for variable length coding   consists of one encoder
\begin{align*}
\phi: & \cX^{n} \times \cY^{n} \longrightarrow \{0,1\}^*
\end{align*}
and two decoders
\begin{align}
 \psi_1: & \phi (\cX^n \times \cY^{n})  \times \cX^{n} \longrightarrow \cZ_1^n, \label{Eq_Dec1} \\
 \psi_2: &  \phi (\cX^n \times \cY^{n}) \times \cY^{n}  \longrightarrow \cZ_2^n. \label{Eq_Dec2}
\end{align}
Here $\{0, 1\}^*$ denotes
the set of all finite length binary sequences and we assume that the encoding 
is prefix free.
Let us define $\hat{Z}_1^n = \psi_1(\phi(X^n,Y^n),X^n) $
and $\hat{Z}_2^n = \psi_2(\phi(X^n,Y^n),Y^n)$.
The probability of error for a $n$ length code  is defined as 
\begin{align}
 P_e^{(n)} \triangleq Pr & \{  (\hat{Z}_1^n,\hat{Z}_2^n)  \neq (Z_1^n,Z_2^n)\} \label{Prob_err}
\end{align}
% Let $(\psi_1(\cdot))_i$ and 
% $(\psi_2(\cdot))_i$ denote
% the $i$-th components of $\psi_1(\cdot)$ and $\psi_2(\cdot)$ respectively.
% A code is called a {\it zero-error code}
% if for each $(X^n, Y^n) \in \cX^n \times \cY^n$, and for all $i=1,2,\cdots,n$,
% \begin{align*}
% & (\psi_1(\phi(X^n,Y^n)),X^n)_i = f(X_i,Y_i)\notag\\
% \mbox{and }\notag\\
% &  (\psi_2(\phi(X^n,Y^n)),Y^n)_i = g(X_i,Y_i)\notag\\
% \end{align*}
% %
The rate of the code is defined as
\begin{eqnarray*}
R & = & \frac{1}{n} \sum_{(x^{n},y^{n})} Pr(x^{n},y^{n}) \mid \phi (x^{n}, y^{n}) \mid,
\end{eqnarray*}
where $\mid \phi (x^{n}, y^{n}) \mid$ denotes the length of
the encoded sequence $\phi (x^{n}, y^{n})$.
A rate $R$ is said to be achievable with zero error if there is a zero-error
code of some length $n$ with rate $R$ and $P_e^{(n)} =0$.
Let $R_0^n$ denote the optimal zero error rate for $n$ length code.
Then the optimal zero error rate $ R^{*}_{0}$ is defined as $R^{*}_{0} = \lim\limits_{n \to \infty}R_0^n$.

% The probability of error is defined as  
% %the probability that either of the decoded messages at decoder 1 and 2  is not equal
% %to the corresponding function values, i.e.,
% \begin{align}
%  P_e^{(n)} \triangleq Pr & \{  (\hat{Z}_1^n,\hat{Z}_2^n)  \neq (Z_1^n,Z_2^n)\} \label{Prob_err}
% \end{align}
% Let $R^n_{vl}$ denote the optimal rate for $n$ instances. Asymptotic optimal rate is 
% denoted by $R^{*}_{vl}$, i.e., $\lim\limits_{n \rightarrow \infty }R^n_{vl} \rightarrow R^{*}_{vl}$.

% Fixed length coding for $n$-length has encoder as
% \begin{align*}
% \phi: & \cX^{n} \times \cY^{n} \longrightarrow \{0,1\}^l
% \end{align*}
% and decoders are similar to variable length coding. The rate of the code is $l/n$ in this case.

A fixed length $(2^{nR},n)$ code consists of one encoder map
\begin{align*}
\phi: & \cX^{n} \times \cY^{n} \longrightarrow \{1,2,\ldots, 2^{nR}\}
\end{align*}
and the two decoder maps as defined in \eqref{Eq_Dec1}, \eqref{Eq_Dec2}. 
% \begin{align*}
%  \psi_1: & \phi (\cX^n \times \cY^{n})  \times \cX^{n} \longrightarrow \cZ_1^n,\\
%  \psi_2: & \phi (\cX^n \times \cY^{n}) \times \cY^{n}  \longrightarrow \cZ_2^n.
% \end{align*}

A rate $R$ is said to be achievable with $\epsilon$-error if there exists a sequence of $(2^{nR},n)$ codes
for which $P_e^{(n)} \rightarrow 0 $ as $n\rightarrow \infty$.
The optimal broadcast rate in this case is
the infimum of the set of all achievable rates
and it is denoted by  $R^*_{\epsilon}$.

 \subsection{Problem formulation: Index coding}

Let $H(i)$ denote the indices of the messages that receiver $i$ has and let 
$X_{H(i)}$ denote their corresponding values. Let us denote the complement of the set $H(i)$ by $H^{c}(i)$.
The set of messages that receiver $i$ has, is denoted by  $X_{H(i)}$. 
The set of messages receiver $i$ wants is $X_{W(i)}$.
% On observing $(X_1^n,X_2^n, \ldots, X_K^n )$, the server broadcasts a message to all $m$ receivers.
% After receiving the message, each receiver decodes  $X_{W(i)}$ using 
% its side information  $X_{H(i)}$. 
For the complementary delivery index coding problem,  $W(i) = H^c(i)$.
The encoder, decoders, probability of error, achievable rate, etc. are
defined similarly as before.

\subsection{Graph theoretic definitions}
\label{sec_graph_def}

Let $G$ be a graph with vertex set $V(G)$ and edge set $E(G)$. A set 
$I \subseteq V(G) $ is called an
independent set if no two vertices in $I$ are adjacent in $G$.
Let $\Gamma(G)$ denote the set of all independent sets of $G$.
A clique of a graph $G$ is a complete subgraph of $G$.
A clique of the largest size is called a maximum clique.
The number of vertices in a maximum clique is called  clique number of $G$ and is denoted by $\omega(G)$.
The chromatic number of $G$, denoted by $\chi(G)$, is the minimum number of colors required 
to color the graph $G$. A graph $G$ is said to be perfect if for any vertex induced
subgraph $G'$ of $G$, $\omega(G') = \chi(G') $.
Note that the vertex disjoint union of perfect graphs is also perfect.

%Let $\chi(G)$ denote the chromatic number of $G$.
%The $n$-fold AND product of $G$, denoted by $G^{\wedge n}$,
%has its vertex set as $V^n$. Two vertices $v^n$ and $v'^n$ are adjacent in $G^{\wedge n}$ if $v_i$ 
%is adjacent or equal to $v_i'$ for all $i \in \{1,\ldots,n\}$. 
The $n$-fold OR product of $G$, denoted by $G^{\vee n}$,
is defined by $V(G^{\vee n}) = (V(G))^n$ and $E(G^{\vee n})=
\{(v^n,v'^n): (v_i,v'_i)\in E(G) \mbox{ for some } i\}$.
The $n$-fold AND product of $G$, denoted by $G^{\wedge n}$,
is defined by $V(G^{\wedge n}) = (V(G))^n$ and $E(G^{\wedge n})=
\{(v^n,v'^n): \mbox{ either } v_i=v'_i \mbox{ or } (v_i,v'_i)\in E(G) \mbox{ for all } i\}$.

For a graph $G$ and a random variable $X$ taking values in $V(G)$,
$(G,X)$ represents a {\it probabilistic graph}.
Chromatic entropy~\cite{Alon_1996} of $(G,X)$ is defined as
\begin{align*}
 H_{\chi}(G,X) &= \mbox{min} \{H(c(X)): \: c  \mbox{ is a coloring of } G \}.
\end{align*}
Let $W$ be distributed over the power set $2^{\cX}$. 
The graph entropy of the probabilistic graph $(G,X)$ is defined as
\begin{align}
 H_G(X) = \min_{X\in W \in \Gamma(G)} I(W;X),
 \label{eq:gentropy}
\end{align}
where $\Gamma(G)$ is the set of all independent sets of $G$. Here the minimum is taken over all 
conditional distributions $p_{W|X}$ which are non-zero only for $X\in W$.
The following result was shown in
\cite{Alon_1996}.
\begin{align}
\lim\limits_{n\to\infty} \frac{1}{n} H_{\chi}(G^{\vee n}, X^n) = H_G(X).
\label{eq:gpentropy}
\end{align}

The complementary graph entropy of $(G,X)$ is defined as 
\begin{align*}
 \bar{H}_{G}(X) = \lim\limits_{\epsilon \to 0} \limsup \limits_{n\to\infty} \frac{1}{n} \log_2 \{ \chi(G^{\wedge n}(T_{P_X, \epsilon }^n) ) \},
\end{align*}
where $T_{P_X, \epsilon }^n$ denotes the $\epsilon$-typical set
of length $n$ under the distribution $P_X$.
It was shown in \cite{Rose_2003} that
\begin{align}
\lim\limits_{n\to\infty} \frac{1}{n} H_{\chi}(G^{\wedge n}, X^n) = \bar{H}_G(X).
\label{eq:gpentropy}
\end{align}

% 
% The definition of graph entropy is extended to the conditional graph entropy to include conditioning \cite{Orlitsky_2001}.
% For a pair of random variables $(X,Y)$ and for a graph $G$ defined on the support set of $X$, 
% the conditional graph entropy of $X$ given $Y$ is defined as 
% \begin{align}
%  H_G(X|Y) = \min_{\substack{ W-X-Y\\X\in W \in \Gamma(G)}} I(W;X|Y).
%  \label{eq_cond_gentropy}
% \end{align}
% where the minimization is over all conditional distribution $p_{W|X} = p_{W|X,Y}$ which is non-zero only for $X\in W$.

To address the function computation problem, we define some suitable graphs.
Let $S_{X^nY^n}$ denote the support set of $(X^n,Y^n)$.
A rook's graph defined over $\cX \times \cY$ has its 
vertex set  $\cX \times \cY$ and edge set
$\{((x,y),(x',y')):x=x' \mbox{ or } y=y', \mbox{ but } (x,y)\neq (x',y')\}$.
For functions $Z_1 =f(X,Y), Z_2 = g(X,Y)$ defined over $\cX \times \cY$, we now define a graph called 
$Z_1Z_2$-modified rook's graph which is similar to the $f$-modified rook's graph defined in \cite{Ravi_2014}. 
%the confusability graph $G_{\chi}f$ has its vector set on $\cX$. 
% Two vertices $x$ and $x'$ in $G_{X}^{f}$ are adjacent if and only if there exists a $y \in \cY$ such 
% that $p(x,y), p(x',y) > 0$ and $f(x,y) \neq f(x',y)$. \\
% $G_{Y}^{f}$ is defined similarly. 

%
\begin{definition}
\label{Def_graph_one}
$Z_1Z_2$-modified rook's graph $G_{XY}^{Z_1Z_2}$ is a subgraph of
the rook's graph on $\cX\times \cY$, which has its vertex set $S_{XY}$, and
two vertices $(x_{1}, y_{1})$ and $(x_{2}, y_{2})$ 
are adjacent if and only if 
%\begin{enumerate}
\begin{align*}
%\item
& \mbox{1) }x_1=x_2 \mbox{ and }  f(x_{1},y_{1}) \neq f(x_{2},y_{2}),\\
\mbox{or }
% \item 
& \mbox{2) } y_1=y_2 \mbox{ and }g(x_{1},y_{1}) \neq g(x_{2},y_{2}). 
\end{align*}
%\end{enumerate}
\end{definition}
\begin{example}
\label{Ex_DSBS}
Let us consider  a {\it doubly symmetric binary source} (DSBS($p$))
$(X,Y)$ where $p_{X,Y}(0,0)= p_{X,Y}(1,1)= (1-p)/2$ and $p_{X,Y}(0,1)= p_{X,Y}(1,0)=p/2$, 
and functions $Z_1,Z_2$ given by
\begin{align}
Z_1 & = X\cdot Y \label{eq2_Z1}\\
Z_2 &= \left\{
\begin{array}{cl}
 Y & \quad \mbox{if} \; Y=0\\
            X & \quad \mbox{if} \; Y =1,
\end{array} \right. \label{eq2_Z2}
\end{align}
$Z_1Z_2$-modified rook's graph of these functions is shown in Fig.~\ref{Mod_Rooks1}.
\end{example}

Next we extend the definition of $G_{XY}^{Z_1Z_2}$ to $n$ instances:
\begin{definition}
\label{Def_graph_multi}
$G_{XY}^{Z_1Z_2}(n)$ has its vertex set $S_{X^nY^n}$, and
two vertices $(x^{n}, y^{n})$ and $(x'^{n}, y'^{n})$ 
are adjacent if and only if 
%\begin{enumerate}
\begin{align*}
%\item
& \mbox{1) }x^n=x'^n \mbox{ and }  f(x_{i},y_{i}) \neq f(x'_{i},y'_{i}) \mbox{ for some } i, \\
\mbox{or }
% \item 
& \mbox{2) } y^n=y'^n \mbox{ and }g(x_{i},y_{i}) \neq g(x'_{i},y'_{i}) \mbox{ for some } i. 
\end{align*}
%\end{enumerate}
\end{definition}

Clearly, $G_{XY}^{Z_1Z_2}(n)$ is the $Z_1^nZ_2^n$-modified rook's graph
on the vertex set $S_{X^nY^n}$. We note here from the definitions that $G_{XY}^{Z_1Z_2}(n)$ is a subgraph 
of $(G_{XY}^{Z_1Z_2})^{\vee n }$.

% \begin{definition}
%  $n$-fold OR product graph of $G_{XY}^{Z_1Z_2}$, denoted by $(G_{XY}^{Z_1Z_2})^{\vee n}$, has its vertex set $ \cX^n \times \cY^n$
%  and two vertices $(x^ny^n), (x'^n,y'^n)$ are connected iff $((x_i,y_i),(x_i',y_i')) \in G_{XY}^{Z_1Z_2}$ for 
%  some $i$.
% \end{definition}

\begin{definition}
 Functions $Z_1,Z_2$ are said to be compatible
 if there exists a function $Z = h(X,Y)$ 
 such that $G_{XY}^{ZZ}=G_{XY}^{Z_1Z_2}$. We call such a graph $G_{XY}^{Z_1Z_2}$  compatible.
\end{definition}

\begin{example}
\label{Ex_DSBS1}
Let us consider another pair $Z_1,Z_2$ which is also defined over a DSBS($p$).
\begin{align}
Z_1 &= \left\{
\begin{array}{cl}
 Y & \quad \mbox{if} \; X=0\\
            X & \quad \mbox{if} \; X =1,
\end{array} \right. \label{eq_Z1}\\
Z_2 & = Y. \label{eq_Z2}
\end{align}

$Z_1Z_2$-modified rook's graph of the above functions is shown in Fig.~\ref{Mod_Rooks2}.
$G_{XY}^{Z_1Z_2}$  in Fig.~\ref{Mod_Rooks2} is not a compatible graph. Whereas 
$G_{XY}^{Z_1Z_2}$  in Fig.~\ref{Mod_Rooks1} is a compatible graph because
it is the same as $G_{XY}^{ZZ}$ for $Z=X\cdot Y$.  
\end{example}

\begin{figure}[h]
 \centering
  \begin{subfigure}[b]{0.18\textwidth}
 \centering
\includegraphics[scale=0.5]{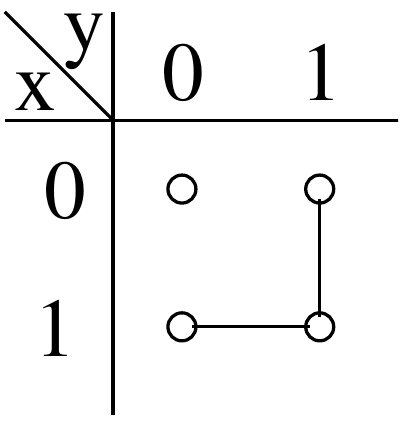}
\caption{$G_{XY}^{Z_1Z_2}$ for $Z_1,Z_2$ defined in \eqref{eq2_Z1},\eqref{eq2_Z2}}
\label{Mod_Rooks1}
\end{subfigure}
\hspace{8 mm}
 \begin{subfigure}[b]{0.18\textwidth}
 \centering
\includegraphics[scale =0.5]{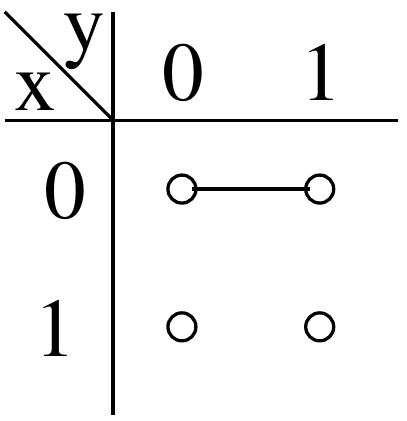}
\caption{$G_{XY}^{Z_1Z_2}$ for $Z_1,Z_2$ defined in \eqref{eq_Z1},\eqref{eq_Z2}}
\label{Mod_Rooks2}
\end{subfigure}

\caption{$Z_1Z_2$-modified rook's graphs}
\label{Fig2}
\end{figure}

\section{Main results}
\label{sec_results}

\begin{figure}[h]
\centering
\includegraphics[width = \columnwidth]{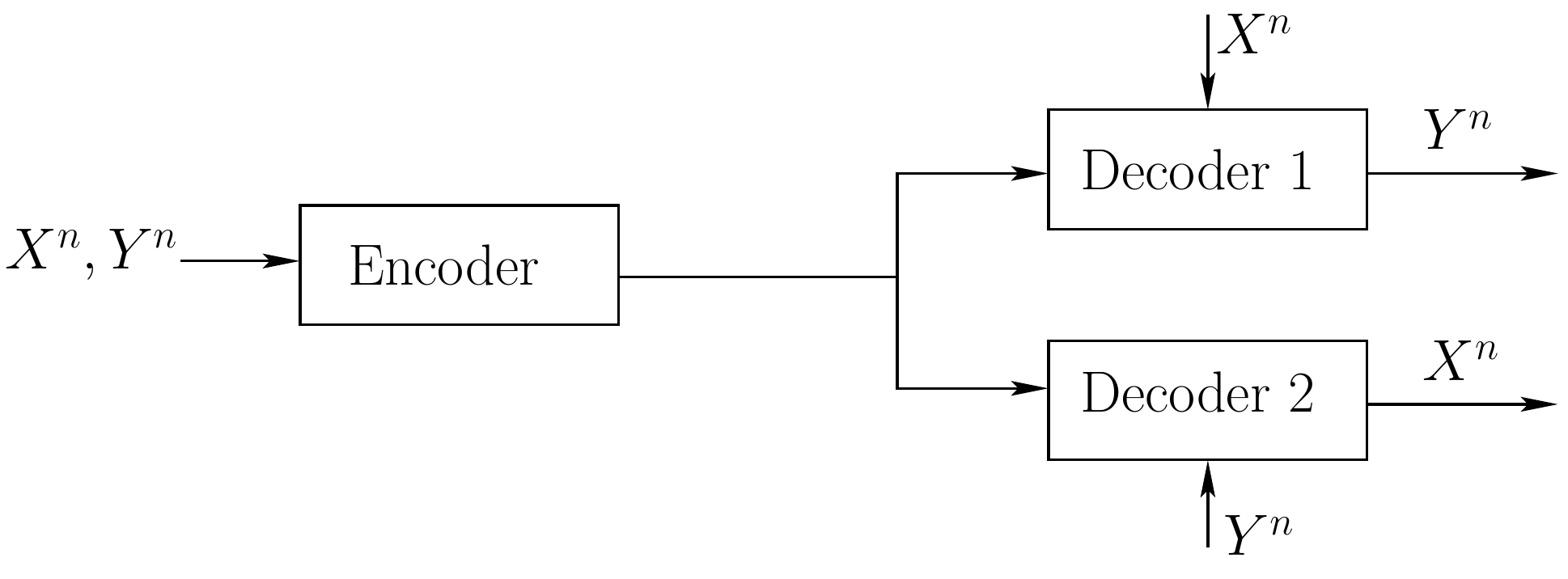}
\caption{Complementary delivery}
\label{Two_Rx}
\end{figure}

Our first result shows that
%Theorem~\ref{Thm_two_Rx} shows that 
the optimal rate for zero error and $\epsilon$-error
are the same for the complementary delivery problem\footnote{ In Section~\ref{sec_proofs} before proving Theorem~\ref{Thm_two_Rx},
we argue that the scheme of binning which achieves the optimal $\epsilon$-error rate 
does not work with zero-error.} shown in Fig~\ref{Two_Rx}.

\begin{theorem}
 \label{Thm_two_Rx}
For the complementary delivery
problem shown in Fig.~\ref{Two_Rx}, 
the optimal zero error broadcast rate
$R^{*}_{0} = \max\{ H(Y|X), H(X|Y) \}$.
\end{theorem}

We now extend Theorem~\ref{Thm_two_Rx} to a more general
complementary delivery index coding problem with arbitrary
number of messages/decoders.

\begin{theorem}
\label{Thm_Many_Rx}
For the complementary delivery index coding problem,
where each receiver demands the complement of its side information,
the optimal zero error broadcast rate
  $R^{*}_{0} =  \max\limits_{i} H(X_{H^{c}(i)}| X_{H(i)})$.
\end{theorem}

% \begin{theorem}
%  \label{Thm_three_messg}
%  For any index coding problem upto three messages, i.e., when $K \leq 3$,
%  optimum broadcast rate for zero error and $\epsilon$-error
%  are the same. 
% \end{theorem}

We now consider broadcast function computation with complementary
side information, and characterize the optimal rate under
$\epsilon$-error in two special cases, and also give single
letter bounds for the optimal rate under $\epsilon$-error and zero error.

\begin{theorem}
\label{Broadcast_rate}
For the broadcast function computation with complementary delivery problem shown in Fig.~\ref{Broadcast_network}
\begin{enumerate}[(i)]
%  \item  \label{Broadcast_rate1} If $Z_1 = Z_2 $, then the optimal rate $R^*_{\epsilon}$ is given by
%  \begin{align*}
%   R^*_{\epsilon} = \max ( H(Z|X), H(Z|Y)). 
%  \end{align*}
 \item  \label{Broadcast_rate2} The optimal rate $R^*_{\epsilon}$ is given by
 \begin{align*}
  R^*_{\epsilon} = \max ( H(Z_1|X), H(Z_2|Y) )
 \end{align*}
 if either of the following conditions hold
  \begin{align*}
  &\mbox{a) } Z_1,Z_2 \mbox{ are compatible. In particular, this condition}\\
  & \quad \mbox{ is satisfied when } Z_1 = Z_2.\\
 &\mbox{b) } X,Y \mbox{ are binary random variables.}
 \end{align*} 
%  \begin{enumerate}[(a)]
%   \item  \label{Broadcast_rate3} $Z_1,Z_2$ are compatible. In particular, this condition is satisfied when $Z_1 = Z_2$.
%   \item  \label{Broadcast_rate4} $X,Y$ are binary random variables.
%  \end{enumerate}
  \item \label{Broadcast_rate3}
Let  
 \begin{align*}
  R_I & = \min_{p(u|x,y)} \max ( I(X;U|Y), I(Y;U|X) ),\\ 
  & \mbox{ where } (X,Y) \in U \in \Gamma(G_{XY}^{Z_1Z_2}).\\
  R_O & = \max_{p(v|x,y)} \max ( I(X;V|Y), I(Y;V|X) ) \\
  & \mbox{ with } \cV| \leq |\cX|.|\cY|+2.
 \end{align*}
 Then $R_O \leq R^*_{\epsilon} \leq R_I$.
 
  \item \label{Broadcast_zero_error}
  The optimal zero error rate $R^*_{0}$ satisfies 
  $ \max\{ H(Z_1|X), H(Z_2|Y) \} \leq R^*_{0} \leq H_{G_{XY}^{Z_1Z_2}}(X,Y)$.
\end{enumerate}
\end{theorem}

\section{Proofs of the results}
\label{sec_proofs}
\subsection{Proof of Theorem~\ref{Thm_two_Rx}}
\label{sec_one_rx}
\begin{remark}
To achieve rates $R$ close to $\max \{H(X|Y),H(Y|X)\}$, let us first consider 
the obvious scheme of random binning $X^n\oplus Y^n$ into $2^{Rn}$ bins.
The decoders can do joint typicality decoding of $X^n\oplus Y^n$
similar to Slepian-Wolf scheme. However, there are two sources
of errors. The decoding errors  for non-typical sequences $(x^n,y^n)$
can be avoided by transmitting those $x^n\oplus y^n$  unencoded, with an 
additional vanishing rate. However, for the same $y^n$, there is
a non-zero probability of two different $x^n\oplus y^n,x'^n\oplus y^n$, both of which are
jointly typical with $y^n$,  being in the same bin; leading to an error
in decoding for at least one of them. It is not clear how to avoid
this type of error with the help of an additional vanishing rate.
\end{remark}

To prove Theorem~\ref{Thm_two_Rx}, we first consider
the problem for single receiver case as shown in Fig.~\ref{One_Rx}.
Witsenhausen \cite{Witsen_1976} studied this problem under
fixed length coding, and gave a single letter characterization
of the optimal rate.
For variable length coding, optimal rate $R^{*}_{0}$ can be argued 
to be $R^{*}_{0} = H(Y|X)$ by using one codebook for each $x$.
Here, we give a graph theoretic 
proof for this, and later extend this technique to prove 
Theorem~\ref{Thm_two_Rx}.

\begin{figure}[h]
\centering
\includegraphics[scale=0.5]{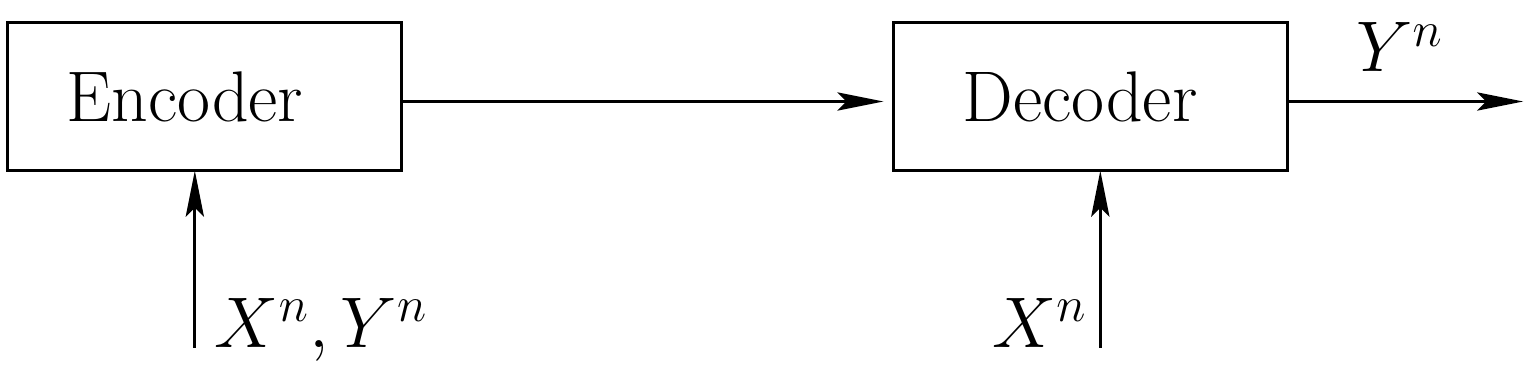}
\caption{One receiver with side information}
\label{One_Rx}
\end{figure}

\begin{lemma}
 \label{Thm_one_Rx}
For the problem depicted  in Fig.~\ref{One_Rx}, $R^{*}_{0} = H(Y|X)$.
\end{lemma}

To prove Lemma~\ref{Thm_one_Rx}, we first prove some claims.
The graph that we use to prove Lemma~\ref{Thm_one_Rx}, is a special case 
of the graph $G_{XY}^{Z_1Z_2}$ defined in Section~\ref{sec_graph_def}, 
obtained by setting $Z_1= Y$ and $Z_2 =\emptyset$. For simplicity, let us denote
this graph by $G$. 
Graph $G$ has its vertex set $S_{XY}$, and
two vertices $(x_{1}, y_{1})$ and $(x_{2}, y_{2})$ 
are adjacent if and only if $x_1=x_2$ and $ y_1 \neq y_2$.
Similarly, we can obtain  the $n$-instance graph for this problem
from Definition~\ref{Def_graph_multi}. For simplicity,
this graph is denoted by $G(n)$.

It is easy to observe that $G$ is  the disjoint union of complete 
row graphs $G_i$
for $i = 1,2, \ldots, |\cX| $, where each $G_i$ 
has vertex set $\{(x_i,y): (x_i,y) \in S_{XY}\}$.

\begin{claim}
\label{Lem_multi_instance}
 For any $n$, the decoder can recover $Y^n$ with zero error if and only if $\phi$ is a 
 coloring of $G(n)$.
\end{claim}

\begin{proof}
 The decoder can recover $Y^n$ with zero error $\Leftrightarrow$ 
 for any $(x^n,y^n),(x^n,y'^n) \in S_{X^nY^n}$ with $y^n \neq y'^n$, $\phi(x^n,y^n) \neq \phi(x^n,y'^n)$
 $\Leftrightarrow$ for any $((x^n,y^n),(x^n,y'^n)) \in E(G(n))$,  $\phi(x^n,y^n) \neq \phi(x^n,y'^n)$
 $\Leftrightarrow$ $\phi$ is a coloring of $G(n)$.
\end{proof}

% The above arguements can be extended to $n$-instances.
% 
% \begin{lemma}
% \label{Lem_multi_instance}
%  For any $n > 1$, node B can recover $Y^n$ with zero error iff $\phi$ is a 
%  coloring of $G^n$.
% \end{lemma}

In the following claim, we identify the vertices of
$G(n)$ with the vertices of $G^{\wedge n}$ by identifying 
$(x^n,y^n)$ with $((x_1,y_1), \ldots, (x_n,y_n))$.

\begin{claim}
\label{Lem_graph_eq}
 $G(n) = G^{\wedge n}$.
\end{claim}

\begin{proof}
For both the graphs, $(x^n,y^n)$ is a vertex if and only if
$p(x_i,y_i) > 0$ for all $i$.
Thus both the graphs have the same vertex set.

Next we show that both the graphs have
the same edge set.
Suppose $(x^n,y^n), (x'^n,y'^n) \in S_{X^nY^n}$ are two distinct pairs.
$\left( (x^n,y^n), (x'^n,y'^n) \right) \in E(G(n))$ $\Leftrightarrow$ $x^n = x'^n$ and  $y^n \neq y'^n$ 
$\Leftrightarrow$  $x_i = x'_i$ for all $i$, and
$y_j \neq y'_j$ for some $j$ $\Leftrightarrow $ for each $i$, either $(x_i,y_i) =(x'_i,y'_i)$ or
$\left( (x_i,y_i),(x'_i,y'_i) \right) \in E(G)$
$\Leftrightarrow \left(((x_1,y_1), \ldots, (x_n,y_n)),((x'_1,y'_1), \ldots, (x'_n,y'_n)) \right) \in E(G^{ \wedge n})$. 
This shows that $G(n) = G^{\wedge n}$.
\end{proof}
\begin{claim}
\label{Lem_AND_multi}
 $R^{*}_{0} = \bar{H}_G(X,Y)$. 
\end{claim}
\begin{proof}
 Claim~\ref{Lem_multi_instance} and the definition of chromatic entropy imply 
 that $ \frac{1}{n} H_{\chi}(G(n),(X^n,Y^n)) \leq R^n_{0} \leq \frac{1}{n} H_{\chi}(G(n),(X^n,Y^n))+ \frac{1}{n} $.
 %Then from Claim~\ref{Lem_graph_eq}, we get 
 %$ \frac{1}{n} H_{\chi}(G^{\wedge n},(X^n,Y^n)) \leq R^n_{0} \leq \frac{1}{n} H_{\chi}(G^{\wedge n},(X^n,Y^n))+ \frac{1}{n} $.
 Using Claim~\ref{Lem_graph_eq}, and taking limit, we get  $R^{*}_{0} 
= \lim\limits_{n \to \infty} \frac{1}{n} H_{\chi}(G^{ \wedge n},(X^n,Y^n))$.
 %It was shown  in \cite{Rose_2003} that for a probabilistic graph $(G,X)$,
 %$\lim\limits_{n \to \infty} \frac{1}{n} H_{\chi}(G^{ \wedge n},X^n) \rightarrow \bar{H}_G(X) $. 
Using \eqref{eq:gpentropy}, this implies $R^{*}_{0} = \bar{H}_G(X,Y)$.
\end{proof}

% We state the following lemma without proof 
% which can be easily shown from the definition of perfect graph.
% 
% \begin{lemma}
%  \label{Lem_disj_union}
%  Disjoint union of perfect graphs is also perfect. 
% \end{lemma}

% \begin{proof}
%  Let us consider a graph $G$ which is a disjoint union of perfect graph $G_i$ for $i =1, \ldots, k$.
%  Let us consider any induced subgraph $G'$ of $G$. We can observe that $G'$ is also disjoint union 
%  of $G'_i$ for $i =1, \ldots, k$. Further, $G'_i \subseteq G_i$. Since all $G_i$s are perfect graphs,
%  $\omega(G'_i) = \chi(G'_i)$. The fact that $G$ is a disjoint union of $G_i$s implies that 
%  $\omega(G') = \max\limits_i \omega(G'_i)$ and $\chi(G') = \max\limits_i \chi(G'_i)$. So we have $\omega(G') = \chi(G')$.
% \end{proof}

\begin{claim}
\label{Lem_graph_perfect}
$G$ is a perfect graph. 
\end{claim}

\begin{proof}
As mentioned before, $G$ is disjoint union of complete graphs. 
Since a complete graph is a perfect graph,
it follows that $G$ is also a perfect graph.
\end{proof}

We now state a lemma from \cite{Korner_1986}.
\begin{lemma}\cite{Korner_1986}
\label{lem_graph_union}
 Let the connected components of the graph $A$ be subgraphs $A_i$. 
 Let $Pr(A_i) = \sum Pr(x), \; x \in V(A_i)$.
 Further, set 
 \begin{align*}
  Pr_i(x) & = Pr(x)[Pr(A_i)]^{-1}, \quad x \in V(A_i).
 \end{align*}
Then $H_A(X) = \sum_i Pr(A_i) H_{A_i}(X_i) $.
\end{lemma}

We now prove Lemma~\ref{Thm_one_Rx}.

% \begin{lemma}
% \label{Lem_graph_cond}
% For any confusion graph $G$, $R^{*}_{0} = H_G(X,Y) = H(Y|X)$. 
% \end{lemma}

{\em Proof of Lemma~\ref{Thm_one_Rx}:}
 For any perfect graph $A$, it is known 
 that $\bar{H}_A(X) = H_A(X)$ \cite{Csiszár_1990,Simonyi_2001}.
 So  Claims~\ref{Lem_AND_multi}
 and \ref{Lem_graph_perfect} imply that $R^{*}_{0} = H_G(X,Y)$. 
 We now use Lemma~\ref{lem_graph_union} to compute $H_G(X,Y)$.
 Recall that each connected component of graph $G$ is a complete graph,
 and the connected component $G_i$, for each $i$, has vertex set 
$\{(x_i,y): (x_i,y) \in S_{XY}\}$
 and $Pr(G_i) = Pr(x_i)$. So we can set the probability of each vertex $(x_i,y) \in G_i$ as
 $Pr(x_i,y)/Pr(x_i)$.
 Since all the vertices in $G_i$ are connected,
 we get $H_{G_i}(x_i,Y) = H(Y|X=x_i)$. Then by using Lemma~\ref{lem_graph_union},
 we get $ H_G(X,Y) = H(Y|X)$.
 This completes the proof of Lemma~\ref{Thm_one_Rx}.
\hfill{\rule{2.1mm}{2.1mm}}

% \subsection{Two receivers with side information}
% \label{sec_two_rx}

Now let us consider the complementary delivery problem shown in Fig~\ref{Two_Rx}. 
This is a special case of the problem shown in Fig.~\ref{Broadcast_network} with
$Z_1 =Y$ and $Z_2 =X$. 
In this case, the $Z_1Z_2$-modified rook's graph
$G_{XY}^{YX}$  has its 
vertex set $S_{XY}$, and two vertices $(x_{1}, y_{1})$ and $(x_{2}, y_{2})$ 
are adjacent if and only if either $x_1=x_2$ and $ y_1 \neq y_2$, or
$y_1=y_2$ and $ x_1 \neq x_2$. 
Now onwards, we denote $G_{XY}^{YX}$ and the $n$-instance graph
$G_{XY}^{YX}(n)$ by $G$  and $G(n)$ respectively.

% \begin{definition}
% Confusion graph
% $G'_{XY}$ has its vertex set $S_{XY}$, and
% two vertices $(x_{1}, y_{1})$ and $(x_{2}, y_{2})$ 
% are adjacent if and only if 
% \begin{enumerate}
%  \item $x_1=x_2$ and $ y_1 \neq y_2$ or
%  \item $y_1=y_2$ and $ x_1 \neq x_2$.
% \end{enumerate}
% \end{definition}
% Let $G'^{n}_{XY}$ be the $n$-instance confusion graph for this problem.

% \begin{definition}
% Graph $G'^n_{XY}$ has vertex set $S_{X^nY^n}$, and
% two vertices $(x^n, y^n$ and $(x'^n, y'^n)$ 
% are adjacent if and only if 
% \begin{enumerate}
%  \item $x^n=x'^n$ and $ y^n \neq y'^n$ or
%  \item $y^n=y'^n$ and $ x^n \neq x'^n$.
% \end{enumerate}
% \end{definition}

% Now onwards, we denote graphs $G'_{XY}$ and $G'^{n}_{XY}$ by $G$  and $(G^n)$ respectively.
% We use Theorem~4 in \cite{Tuncel_2009} to prove Theorem~\ref{Thm_two_Rx}. We state that particular theorem
% below.

We now state a Theorem from \cite{Tuncel_2009} which is used to prove
Theorem~\ref{Thm_two_Rx}.

\begin{theorem}\cite{Tuncel_2009}
\label{Thm_and_union}
Let $\cG = (G_1,\ldots, G_k)$ be a family of graphs on the same vertex set.
If $R_{\min}(\cG, P_X) := \lim\limits_{n \to \infty} \frac{1}{n} \left( H_{\chi}(\bigcup_i G_i^{\wedge n}, P_X^n) \right)$,
then $R_{\min}(\cG, P_X) = \max\limits_{i} R_{\min}(G_i,P_X)$ where $R_{\min}(G_i,P_X) = \bar{H}_{G_i}(X)$.
\end{theorem}

We are now ready to prove Theorem~\ref{Thm_two_Rx}.

{\em Proof of Theorem~\ref{Thm_two_Rx}:}
%The proof is very similar to the proof of 
%{\cite [Theorem~1] {Simonyi_2003}}.
%This proof was extended 
%to variable length coding in \cite{Tuncel_2009} to prove Theorem~\ref{Thm_and_union}.
For $i=1,2$, let $G_i$ be the modified rook's graphs corresponding to 
decoding with side information at decoder $i$. 
So the modified rook's graph for
the problem with two decoders is given by $G= G_1 \bigcup G_2 $.
Two vertices $(x^n,y^n)$ and $(x'^n,y'^n)$ are connected in the
corresponding $n$ instance graph $G(n)$ if and only if they are connected
either in $G_1(n)$ or in $G_2(n)$. This implies that
$G(n)= G_1(n) \bigcup G_2(n)$. 
This shows that both the decoders can decode with zero error if and only if $\phi$
is a coloring of $G(n)$. This fact and the definition of chromatic entropy imply that
$R^{*}_{0} = \lim\limits_{n \to \infty} \frac{1}{n} H_{\chi} \left(G(n),(X^n,Y^n) \right)$.
From Claim~\ref{Lem_graph_eq}, it follows that
$G(n)=G^{\wedge n}_1 \bigcup G^{\wedge n}_2$. Then by using Theorem~\ref{Thm_and_union}, we get
$R^{*}_{0} = \max\{ \bar{H}_{G_1}(X,Y), \bar{H}_{G_2}(X,Y) \} $.
As argued in the proof of Lemma~\ref{Thm_one_Rx}, $\bar{H}_{G_1}(X,Y)
=H(Y|X)$ and $\bar{H}_{G_2}(X,Y)=H(X|Y)$.
Thus $R^{*}_{0} =\max\{ H(Y|X), H(X|Y) \}$.
\hfill{\rule{2.1mm}{2.1mm}}

\subsection{Proof of Theorem~\ref{Thm_Many_Rx}}

The proof of Theorem~\ref{Thm_Many_Rx} follows 
by the same arguments as that of Theorem~\ref{Thm_two_Rx}, and is thus
omitted.

\subsection{Proof of Theorem~\ref{Broadcast_rate}}

% The proof of part~(\ref{Broadcast_rate1}) follows as a special case of Theorem~\ref{Broadcast_rate}, part~(\ref{Broadcast_rate2}).
% So we omit the proof here.
%\begin{proof}
%  Converse for this rate follows from the cut set bound. To achieve this rate, encoder first
%  computes $Z^n$, then uses Slepian-Wolf coding to compress it at a rate $ \max ( H(Z|X), H(Z|Y))$. 
%  Then both the decoders can compute $Z^n$ with negligible probability of error.
%\end{proof}
%
\comment{When $Z_1,Z_2$ are compatible or $X,Y$ are  binary random variables,  
the optimal rate  $R^*$ is given in  part~(\ref{Broadcast_rate2}).}	
Lemma~\ref{lemma_rooks} below is used in the achievability proof of part~(\ref{Broadcast_rate2}).
\begin{lemma}
\label{lemma_rooks}
If $Z_1,Z_2$ are compatible such that $G_{XY}^{ZZ}=G_{XY}^{Z_1Z_2}$ for $Z=h(X,Y)$, then
 $H(Z_1|Z,X) = 0$ and  $H(Z_2|Z,Y) = 0$. As a consequence,
  $H(Z|X) = H(Z_1|X)$ and  $H(Z|Y) = H(Z_2|Y)$.
\end{lemma}
\begin{proof}
 For any $(x,y)$ and $(x,y')$ , observe that 
 \begin{align}
  h(x,y) & = h(x,y') \iff f(x,y) = f(x,y') \label{fun_ZZ1}.
 \end{align}
 Similarly, for any $(x,y)$ and $(x',y)$,
 \begin{align}
  h(x,y) & = h(x',y) \iff g(x,y) = g(x',y) \label{fun_ZZ2}
 \end{align} 
For a given $X=x$ and $Z=h(X,Y)=z$, let us consider the set
of possible $y$, $A_{x,z}=\{y': h(x,y')=z \}.$
% \begin{align*}
% A_{x,z}=\{y': h(x,y')=z \}.
% \end{align*}

By \eqref{fun_ZZ1}, $f(x,y')=f(x,y'')\,\, \forall y',y''\in A_{x,z}$.
Thus, denoting this unique value by $z_1:=f(x,y')$, we have
$Pr\{Z_1=z_1|X=x,Z=z\}=1$. So we have
 $H(Z_1|Z,X)=0 $ and similarly $ H(Z_2|Z,Y)=0$.
Using similar lines of arguments, we get
  $H(Z|Z_1,X)=0$ and $H(Z|Z_2,Y)=0$.
 Then we get the following.
 \begin{align*}
 H(Z|X)   & = H(Z|X)+H(Z_1|Z,X)\\
         & = H(Z_1,Z|X)\\
         & = H(Z_1|X) + H(Z|Z_1,X)\\
         & = H(Z_1|X),
 \end{align*}
% where $(a),(b)$ follows from \eqref{fun_ZZ1}.\\
Similarly, we get  $H(Z|Y) = H(Z_2|Y)$.
\end{proof}

% \begin{theorem}
% \label{Broadcast_rate1}
% The optimal rate $R^*$ is given by
%  \begin{align*}
%   R^* = \max ( H(Z_1|X), H(Z_2|Y) )
%  \end{align*}
%  if
%  \begin{enumerate}[(i)]
%   \item $Z_1,Z_2$ are $Z$-satisfiable. \label{Thm1_Part1}\\
%   or
%   \item $X,Y$ are binary random varaibles. \label{Thm1_Part2}
%  \end{enumerate} 
% \end{theorem}

{\em Proof of part~(\ref{Broadcast_rate2}):} We first prove part~(\ref{Broadcast_rate2})~a).
 Converse for $R^*$ follows from the cut-set bound. Now let us consider the achievability of $R^*_{\epsilon}$.
 The encoder first computes $h(x^n,y^n)$ and then
uses Slepian-Wolf binning to compress it at a rate $\max(H(Z|X), H(Z|Y ))$. Then decoder 1 and 2
can compute $Z^n$ with negligible probability of error.
From Lemma~\ref{lemma_rooks}, it follows that encoder 1 can recover $Z_1^n$ from $Z^n$ and $X^n$.
Similarly, encoder 2 computes $Z_2^n$ from $Z^n$ and $Y^n$. From Lemma~\ref{lemma_rooks}, 
we have $\max(H(Z|X), H(Z|Y )) = \max ( H(Z_1|X), H(Z_2|Y) )$.
When $Z_1=Z_2=Z$, from
the above arguments it is easy
to see that $\max(H(Z|X), H(Z|Y))$ is achievable.

Now let us consider part~(\ref{Broadcast_rate2})~b).
% \noindent
% \underline{ {\em Proof of part~(\ref{Broadcast_rate2})~b):}} 
Here also converse for $R^*_{\epsilon}$ follows from the cut-set bound. For achievability, let us consider
$G_{XY}^{Z_1Z_2}$. When $X,Y$ are binary random variables, any $G_{XY}^{Z_1Z_2}$ is a subgraph 
of the ``square'' graph with four edges. When $S_{XY} = \cX \times \cY$, if graph $G_{XY}^{Z_1Z_2}$
has one edge then $Z_1,Z_2$ are not compatible. 
It can be checked that any other possible graph $G_{XY}^{Z_1Z_2}$ is compatible.
For those compatible graphs, the proof follows from part~(\ref{Broadcast_rate2})~a).
For a graph with only one edge, w.l.o.g., let us consider the graph shown in Fig.~\ref{Mod_Rooks2}.
It is clear that $H(Z_2|Y) =0$ and so decoder 2 can recover $Z_2$ only from $Y$. 
For decoder 1, we need an encoding rate $R=H(Z_1|X)$. Thus the rate $\max(H(Z_1|X), H(Z_2|Y ))=H(Z_1|X)$
is achievable.
\hfill{\rule{2.1mm}{2.1mm}}
% Any other $Z_1,Z_2$ are compatible for which the achievability follows
% from part~(\ref{Broadcast_part2a}). Now let us consider  $Z_1,Z_2$ for which $G_{XY}^{Z_1Z_2}$ has one edge.
% For such functions, it is easy to see that either $H(Z_2|Y)=0$ (where single edge is horizontal as shown in Fig.~\ref{Mod_Rooks1}) or $H(Z_1|X)=0$. W.l.o.g., let us assume that
% $H(Z_2|Y)=0$. It means that there is one horizontal edge in $G_{XY}^{Z_1Z_2}$. Then the encoder computes $z_1^n$
% and does binning at a rate $H(Z_1|X)$. Then decoder 1 can compute $Z_1^n$ with negligible probability error
% while decoder 2 can compute $Z_2^n$ from $Y^n$ itself. This shows thet the rate $\max ( H(Z_1|X), H(Z_2|Y) )$ is achievable.
% \end{proof}
%
% In the achievability of $R^*$, we use the fact that $(X,Y)$ is available to the encoder.
% It can be shown that $H(Z|Z_1,Z_2) \neq 0$. So unlike Gray-Wyner network, 
% it is not known whether the rate $R^*$ is achievable
% for an encoder which observes only $Z_1,Z_2$ instead of $X,Y$.
% %

Before proving part~(\ref{Broadcast_rate3}) of Theorem~\ref{Broadcast_rate},
we present a useful lemma.
\begin{lemma}
\label{lemma2_rooks}
 Let $ W \in \Gamma(G_{XY}^{Z_1Z_2})$ be a random variable such that $(X,Y) \in W $.
 Then $H(Z_1|W,X)=0$ and $H(Z_2|W,Y)=0$.
\end{lemma}
\begin{proof}
 Since $w$ is an independent set of $G_{XY}^{Z_1Z_2}$, for each $x\in \cX$, $f(x,y') = f(x,y'')$
 for all $(x,y'),(x,y'') \in w$. So decoder 1 can compute $f(x,y)$
 from $(w,x)$ whenever $p(w,x,y)>0$. Similarly, decoder 2 can compute $g(x,y)$ from $(w,y)$ 
 whenever $p(w,x,y)>0$. This implies that $H(Z_1|W,X)=0$ and $H(Z_2|W,Y)=0$.
\end{proof}

Given $x$ and independent set $w$, since the value of $z_1$ is unique, this unique 
value is denoted by $z_1(w,x)$ with abuse of notation.

%
% \begin{theorem}
% \label{Broadcast_rate2}
% Let  
%  \begin{align*}
%   R_I & = \min_{p(u|x,y)} \max ( I(X;U|Y), I(Y;U|X) ),\\ 
%   & \mbox{ where } (X,Y) \in U \in \Gamma(G_{XY}^{Z_1Z_2}).\\
%   R_O & = \max_{p(v|x,y)} \max ( I(X;V|Y), I(Y;V|X) ) \\
%   & \mbox{ with } \cV| \leq |\cX|.|\cY|+2
%  \end{align*}
%  Then $R_O \leq R^* \leq R_I$.
% %   \begin{enumerate}[(i)]
% %   \item $R^* \leq R_I$ \label{Thm2_part1}
% %   \item  $R^* \geq R_O$. \label{Thm2_part2}
% %  \end{enumerate} 
%  
% \end{theorem}
%
{\em Proof of part~(\ref{Broadcast_rate3}):}
% {\it Proof of part~(\ref{Thm2_part1})}:
First we prove $ R^*_{\epsilon} \leq R_I$.
Let $U$ be a random variable such that it satisfies the conditions of $R_I$ in part~(\ref{Broadcast_rate3}).

{\em Generation of codebooks:} 
Let  $\{ U^n(l) \},l \in [ 1:{2^{n\tilde{R}}}]$, be a set of sequences, each chosen i.i.d. according
to $\prod_{i=1}^n p_{U}(u_{i})$.
Partition the set of sequences $U^n(l)$, $l \in [1:2^{n\tilde{R}} ]$,  into equal-size 
bins, $B(m) = [(m-1)2^{n(\tilde{R}-R)}+1: m 2^{n(\tilde{R}-R)}]$, where $m \in [1: 2^{nR}]$.
% Encoder 0 first compute $Z^n$ from $(X^n,Y^n)$ and randomly bin all the $Z^n$ 
% sequences into $2^{nR}$ bins. The binning is done independently and uniformly distributed on 
% $\{1,2,\ldots,2^{nR} \}$. Let $B(m)$ denote the the set of $Z^n$ sequences alloted to bin $m$.

{\em Encoding:} 
Given $(x^n,y^n)$, the encoder finds an index $l$ such that  $(x^n,y^n,u^n(l)) \in T_{\epsilon}^n(X,Y,U)$.
If there is more than one such index, it selects one of them uniformly at random. If there is no
such index, it selects an index uniformly at random from $[1:{2^{n\tilde{R}}}]$. The encoder sends
the bin index $m$ such that $l \in B(m)$. 

{\em Decoding:} 
% Let us look at the decoding operation of decoder 1.
Once decoder 1 receives the message from the encoder, it finds the unique index $\hat{l} \in B(m)$
such that $(x^n,u^n(\hat{l})) \in  T_{\epsilon}^n(X,U)$. If there is no unique $\hat{l} \in B(m)$,
it sets $\hat{l}=1$.
It then computes the function values $z_{1i}$ as $\hat{z}_{1i} = z_{1i}(u_i(\hat{l}), x_i)$ for $i \in [1;n]$.
Decoder 2 operates similarly.
% The decoder 1 looks for a unique $Z^n \in B(m_0)$  such that $(Z^n, U^n(m_1)) \in T_{\epsilon}^n(Z,U)$.
% If there is more than one, it declares an error. Similar decoding operation is performed by decoder 2. 

{\em Analysis of error:} Let $(L,M)$ denote the chosen codeword and bin indices at encoder and let $\hat{L}$ be the
index estimate given by decoder 1. 
% From Lemma~\ref{lemma_rooks}, we get that if $ \hat{L} = L $, decoder 1 can compute
% $Z_1^n$ with no error. 
Decoder 1 makes an error if and only if the following event $\cE_1$ happens.
\begin{align*}
\cE_1 & = \{ (U^n(\hat{L}), X^n, Y^n) \notin  T_{\epsilon}^n)\}
\end{align*}
Event $\cE_1$ happens only if one of the following events happens.
\begin{align*}
\cE_{11} & = \{ (U^n(l), X^n, Y^n) \notin  T_{\epsilon'}^n) \mbox{ for all } l \in [1: 2^{n\tilde{R}}] \}\\
\cE_{12} & = \{  \exists \; \tilde{l} \neq L \mbox{ such that } \tilde{l}  \in B(M), (U^n(\hat{l}), X^n) \in  T_{\epsilon}^n\}  
\end{align*}
Under $\cE_{11}^c$, if $ \hat{L} = L $, then decoder 1 can compute
$Z_1^n$ with no error. 
The probability of error for decoder 1 is upper bounded as $ P(\cE_1) \leq  P(\cE_{11})+P(\cE_{12}).$

By covering lemma \cite{Elgamal_Kim}, $P(\cE_{11}) \rightarrow 0 $ as $n \rightarrow \infty $ if
$\tilde{R} > I(X,Y;U) + \delta(\epsilon')$.  $P(\cE_{12})$ is the same as the probability of error $P(\cE_{3})$
in {\cite [Lemma 11.3] {Elgamal_Kim}} if we replace $Y^n$ with $X^n$. By packing lemma, $P(\cE_{12}) \rightarrow 0$ if
$\tilde{R}-R < I(X;U) - \delta(\epsilon) $. Combining these two bounds, we get $ P(\cE_1) \rightarrow 0$
as $n \rightarrow \infty $ if $R> I(X,Y; U) - I(X;U) + \delta(\epsilon)+\delta(\epsilon')$.
This shows that any rate $R > I(U;Y|X)$ is achievable for decoder 1.

Similarly for decoder 2, any rate $R > I(U;X|Y)$ is achievable under the same encoding. So we get that $R > \max ( I(X;U|Y), I(Y;U|X) )$
is an achievable rate. 
Now we show that $R_O \leq R^*_{\epsilon}$. 
\begin{align*}
 nR & \geq H(M)\\
 & \geq H(M|X^n)\\
 &= I(M;Y^n|X^n) \; ( M  \mbox{ is a function of } (X^n,Y^n))  \\
%  & \stackrel{(a)}{\geq} H(Z_1^n, M|M_1) - n\epsilon_n\\
% & = H(Y^n|X^n) - H(Y^n|X^n,M)\\
%  & = \sum_{i=1}^{n} H(Y_{i}|X_i) -\sum_{i=1}^{n} H(Y_{i}|Y^{i-1},X^n,M)\\
  & \stackrel{(a)}{\geq} \sum_{i=1}^{n} H(Y_{i}|X_i) -\sum_{i=1}^{n} H(Y_{i}|Y^{i-1},X_i,X^{i-1},M)\\
 & = \sum_{i=1}^{n} I (Y_i;V_i|X_i) \mbox{ (where }V_i = (M,X^{i-1},Y^{i-1})),
%   & \geq \sum_{i=1}^{n} H(Z_{1i}|M_1,Y^{i-1},X^{i-1}) - n\epsilon_n\\
%  & = \sum_{i=1}^{n} H(Z_{1i}|U_i) - n\epsilon_n\\ 
%  & \geq I(Z_1^n,Z_2^n;M)\\
%  & = \sum_{i=1}^{n} I(Z_{1i},Z_{2i}; M|Z_1^{i-1},Z_2^{i-1})\\
%  & = I(Z_{1i},Z_{2i}; M,Z_1^{i-1},Z_2^{i-1})\\
%  & = I(Z_{1i},Z_{2i};W_i)
\end{align*}
where $(a)$ follows from the fact that conditioning reduces entropy.
Now defining a timesharing random variable $Q,V=(V_Q,Q), 
X_{Q} = X$ and $Y_{Q} = Y$; and using support lemma, the result follows.
\comment{
Let $Q$ be such that $Pr(Q=i) = \frac{1}{n}$. Then we get
\begin{align*}
 R & \geq \frac{1}{n} \sum_{i=1}^{n} I(Y_i;V_i|X_i,Q=i) =   I(Y_Q;V_Q|X_Q,Q)
\end{align*}
Since $Q$ is independent of $(X_Q,Y_Q)$,
\begin{align*}
I(Y_Q;V_Q|X_Q,Q) =  I(Y_Q;V_Q,Q|X_Q)
\end{align*}
By defining $V=(V_Q,Q), X_{Q} = X$ and $Y_{Q} = Y$, we get $R \geq I(Y;V|X)$.
By similar lines of argument, we have  $R \geq I(X;V|Y)$. From the support lemma,
we can show that $|\cV| \leq |\cX|.|\cY|+2$. 
}

% {\em Proof of Part~(\ref{Broadcast_zero_error}):}
%  $ \max\{ H(Z_1|X), H(Z_2|Y) \}$ is a lower bound for $R^*_{\epsilon}$ from the cut-set bound. This implies that
%   $ \max\{ H(Z_1|X), H(Z_2|Y) \}$  is also a lower bound for $ R^*_{0}$. Next we show that  
%   $R^*_{0} \leq H_{G_{XY}^{Z_1Z_2}}(X,Y)$. From the definition of $G_{XY}^{Z_1Z_2}(n)$, it can be observed that
%   the decoders 1 and 2 can compute functions $Z_1$ and $Z_2$ respectively if and only if $\phi$
%   is a coloring of  $G_{XY}^{Z_1Z_2}(n)$. This shows 
%   that $R^*_{0} = \lim\limits_{n \to \infty} H_{\chi}\left(G_{XY}^{Z_1Z_2}(n), (X^n,Y^n)\right)$.
%   To get an upper bound for this limit, we consider the $n$-fold OR
%   product graph $(G_{XY}^{Z_1Z_2})^{\vee n }$.
%   As noted after Definition~\ref{Def_graph_multi}, $G_{XY}^{Z_1Z_2}(n)$ is a subgraph 
%   of $(G_{XY}^{Z_1Z_2})^{\vee n }$. Thus 
%   $R^*_{0} \leq \lim\limits_{n \to \infty} H_{\chi}\left((G_{XY}^{Z_1Z_2})^{\vee n}, (X^n,Y^n)\right) = H_{G_{XY}^{Z_1Z_2}}(X,Y)$.
\hfill{\rule{2.1mm}{2.1mm}}

% \section{Index Coding}
% \label{sec_index}
% \input{index.tex}
% 
%   % 
% \section{When $Z_1 \neq Z_2$}
% \label{sec_diff}
% \input{Diff_func.tex}

% \section{Conclusion}
% \label{sec_conclusion}
% \input{Conclusion.tex}

% \balance
\section*{Acknowledgment}
The work  was supported in part by the Bharti Centre for Communication, IIT Bombay,
a grant from the Information Technology Research Academy, Media Lab Asia, to IIT Bombay,
and a grant from the Department of Science \& Technology to IIT Bombay.

% \appendices
% \section{Proof of Theorem~\ref{Thm_Many_Rx}}
% \input{append1.tex}

\end{document}